\pdfoutput=1
\RequirePackage{etex}
\documentclass{article}
\title{Towards Optimal Subsidy Bounds for Envy-freeable Allocations}
\author{
  Yasushi Kawase\thanks{The University of Tokyo, \texttt{kawase@mist.i.u-tokyo.ac.jp}}
  \and
  Kazuhisa Makino\thanks{Kyoto University, \texttt{makino@kurims.kyoto-u.ac.jp}}
  \and
  Hanna Sumita\thanks{Tokyo Institute of Technology, \texttt{sumita@c.titech.ac.jp}}
  \and
  Akihisa Tamura\thanks{Keio University, \texttt{aki-tamura@math.keio.ac.jp}}
  \and
  Makoto Yokoo\thanks{Kyushu University, \texttt{yokoo@inf.kyushu-u.ac.jp}}
}
\date{}

\usepackage{fullpage}
\usepackage[numbers,sort]{natbib}
\usepackage[colorlinks,citecolor=blue,linkcolor=blue]{hyperref}

\usepackage{comment}
\usepackage{amsmath,amsthm,amssymb,mathtools,bm}
\usepackage{nicematrix}
\NiceMatrixOptions{code-for-first-row = \scriptstyle\color{gray}, code-for-first-col = \scriptstyle\color{gray}}

\usepackage{cleveref,autonum}
\newtheorem{theorem}{Theorem}

\newtheorem{example}{Example}
\newtheorem{lemma}{Lemma}

\usepackage{tikz}
\usetikzlibrary{arrows.meta}
\usepackage{enumitem}

\newcommand{\id}{\mathtt{id}}

\newcommand{\bX}{\bm{X}}
\newcommand{\bA}{\bm{A}}

\newcommand{\bp}{\bm{p}}
\newcommand{\bq}{\bm{q}}

\DeclareMathOperator*{\argmax}{arg\,max}

\begin{document}

\maketitle

\begin{abstract}
We study the fair division of indivisible items with subsidies among $n$ agents, where the absolute marginal valuation of each item is at most one.
Under monotone valuations (where each item is a good), Brustle et al.~\cite{Brustle2020} demonstrated that a maximum subsidy of $2(n-1)$ and a total subsidy of $2(n-1)^2$ are sufficient to guarantee the existence of an envy-freeable allocation.
In this paper, we improve upon these bounds, even in a wider model.
Namely, we show that, given an EF1 allocation, we can compute in polynomial time an envy-free allocation with a subsidy of at most $n-1$ per agent and a total subsidy of at most $n(n-1)/2$.
Moreover, we present further improved bounds for monotone valuations.
\end{abstract}

\section{Introduction}
We consider the problem of \emph{fairly} dividing items
among agents. The notion of fairness that has been extensively
studied in the literature is \emph{envy-freeness}~\citep{Foley}. It
requires that no agent wants to swap her bundle with another agent's. When the items to be allocated are divisible, the classical result ensures the existence of an envy-free allocation~\citep{Varian}.
In contrast, when the items are indivisible, envy-freeness is not
a reasonable goal. For instance, consider $n$ agents with $n\ge 2$ and a single item valued at $1$ by each agent. Allocating the only item to an agent results in envy from the other agents, as they get nothing. Thus, envy-free allocations may not exist when the items are indivisible.

One way to circumvent this issue is to relax the fairness requirement.
For example, \emph{envy-freeness up to one item} (EF1) requires that
when agent $i$ envies agent $j$, the envy can be eliminated by either
(i) removing one item from agent $j$'s bundle, or 
(ii) removing one item from agent $i$'s bundle. 
It is known that an EF1 allocation is guaranteed to exist if each item is either a good 
or a chore 
for any agent, i.e., doubly monotone valuations~\cite{LiptonMaMo04,bhaskar+2021}. 
Moreover, for general valuations, the existence of an EF1 allocation is assured when there are only $2$ agents~\cite{BBB+2020}.

Another way to circumvent this issue is monetary compensation (subsidy). 
Since money is divisible, it can be a powerful tool to achieve envy-freeness.
However, since the subsidy payments must be provided by
an external agent (e.g., a government or a funding agency),
it is desirable that the total subsidy amount is bounded. 
Thus, in this paper, we study the fair division of indivisible items
with limited subsidies. 

Most of the existing works on the fair division of indivisible items with limited subsidies
focus on some special cases.
For example, \citet{Maskin1987} and \citet{Klijn}
consider the case that 
the number of agents and the number of items are equal and
each agent can be allocated at most one item.
\citet{HalpernShah2019} consider an extended model where
the number of agents and number of items may differ,
and each agent can be allocated more than
one item, assuming the valuation of each agent is binary additive.
\citet{babaioff2020fair} and \citet{aamasversion}
consider the case that the valuation of each agent
is matroidal (which is not necessarily additive). 
\citet{Barman2023} examine a broader class of valuations in which the marginal valuation of each item is binary. 
As far as the authors are aware, the most general model considered 
so far is monotone valuations~\cite{Brustle2020}, 
where the marginal contribution of each item is non-negative. 

In this paper, we study the fair division of indivisible items with limited subsidies when the valuations are not restricted to be monotone.
We assume that the valuations are normalized so that the absolute marginal value of each item is at most $1$ (i.e., between $-1$ and $1$). 

For monotone valuations, which are special cases of our model, 
\citet{Brustle2020} show that envy-free allocation always exists with a subsidy
of amount at most $2(n-1)$ per agent, and the total amount is $2(n-1)^2$.
However, the only known lower bound on the total subsidy is $n-1$ (which can be obtained using the case
with $n$ agents and one item described at the beginning of this section), 
it remains an open question whether 
we can improve on the total subsidy bound of $2(n-1)^2$ for monotone valuations, 
as mentioned in a recent survey paper \cite[Open Question 9]{liu2023mixed}.

\subsection*{Our Results}
In this paper, we present improved upper bounds for the subsidies necessary to achieve envy-freeness.
We demonstrate that, given an EF1 allocation, an envy-free allocation with a subsidy can be constructed in polynomial time where each agent receives a subsidy of at most $n-1$ and $n(n-1)/2$ in total (\Cref{thm:EF1-upper}).
This result implies that, when valuation functions are doubly monotone or there are only two agents, such envy-free allocations with limited subsidies can be computed in polynomial time. Notably, this improves upon the necessary subsidy amount for the existing case of monotone valuations, as monotone valuations are also doubly monotone. 
Furthermore, when $n=2$, our obtained bounds are best possible since a subsidy of $1$ is indispensable. We also show that from an EF$k$ allocation (i.e., pairwise envy can be eliminated by removing at most $k$ items), we can construct an envy-free allocation with a subsidy of at most $k(n-1)$ per agent and a total subsidy of $k\cdot n(n-1)/2$.

It is worth mentioning that our upper bounds of $n-1$ per agent and $n(n-1)/2$ in total cannot be improved when considering an arbitrary EF1 allocation (\Cref{ex:n-1}).
We overcome this impossibility by making a slight modification of the bundles.
To be exact, for three or more agents with monotone valuations, we improve the bounds further to $n-1.5$ per agent and $(n^2-n-1)/2$ in total (\Cref{thm:improved}).

\subsection*{Related Work}
The concept of compensating an indivisible resource allocation with money
has been prevalent in classical economics literature~\citep{Alkan1991,Maskin1987,Klijn,moulin2004fair,SunYang2003,Svensson1983,Tadenuma1993}.
Much of the classical work has focused on the unit-demand case in
which each agent is allocated at most one good. Examples include the
famous rent-division problem of assigning rooms to several housemates
and dividing the rent among them \citep{SuRentalHarmony}. It is known
that, for a sufficient amount of subsidies, an envy-free allocation
exists~\citep{Maskin1987} and can be computed in polynomial time~\citep{Aragones,Klijn}.

Most classical literature, however, has not considered a situation in which the number of items to be allocated exceeds the number of agents, in contrast to the rich body of recent literature on the multi-demand fair division problem. 
\citet{HalpernShah2019} recently extended the model to the multi-demand setting wherein multiple items can be allocated to one agent. 
Despite the existence of numerous related papers,
\citet{HalpernShah2019} is the first work to study the asymptotic bounds on the amount of subsidy required to achieve envy-freeness. 
They showed that an allocation is envy-freeable with a subsidy if and only if the agents cannot increase the utilitarian social welfare by permuting bundles. This characterization implies that an allocation that can be made envy-free with a subsidy needs to satisfy some efficiency condition; hence, an approximately fair allocation, such as an EF1 allocation~\citep{Budi11a}, may not be an envy-freeable allocation.
It was conjectured in~\cite{HalpernShah2019} that, for additive valuations in which the value of each item is at most $1$, giving at most $1$ to each agent is sufficient to eliminate envies. 
\citet{Brustle2020} affirmatively settled this conjecture
by designing an algorithm that iteratively solves a maximum-matching instance.

\citet{babaioff2020fair} and \citet{benabbou2020} studied the fair allocation of indivisible items with matroidal valuations. The prioritized egalitarian mechanism proposed by \citet{babaioff2020fair} returns an allocation that maximizes the Nash welfare and achieves envy-freeness up to any good and utilitarian optimality. 
\citet{aamasversion} developed a strategy-proof, 
polynomial-time implementable mechanism called
\emph{subsidized egalitarian} mechanism, which requires a subsidy of the amount
at most $1$ per agent and $n-1$ in total. 
Furthermore, \citet{BarmanKrishna2022IJCAI} examined a more general model
with dichotomous marginals and obtained the same bounds. 

\citet{caragiannis2020computing} studied the computational complexity of approximating the minimum amount of subsidies required to achieve envy-freeness. \citet{Aziz2020} considered another fairness requirement, the so-called \emph{equitability}, in conjunction with envy-freeness and characterized an allocation that can be made both equitable and envy-free with a subsidy. \citet{narayan2021birds} studied a related but different setting with transfer payments; they analyzed the impact of introducing some amount of transfers on the Nash welfare and utilitarian welfare while achieving envy-freeness.

\section{Preliminaries}
We model fair division with a subsidy as follows.
For each natural number $k \in \mathbb{N}$, we denote $[k]=\{1,\ldots,k\}$.
Let $N=[n]$ be the set of $n$ agents and let $M=\{e_1,\dots,e_m\}$ be the set of $m$ indivisible items.
Each agent $i$ has a \emph{valuation function}, denoted as $v_i\colon 2^M \to \mathbb{R}$, where $\mathbb{R}$ represents the set of real numbers.
We assume that the functions $v_i$'s are given as value oracles. 
In addition, we assume that the maximum marginal contribution of each item is at most one, i.e., 
$|v_i(X\cup \{e\})-v_i(X)|\le 1$ for any $i\in N$, $e\in M$, and $X\subseteq M\setminus\{e\}$.

We define an item $e\in M$ as a \emph{good} for agent $i\in N$ if $v_i(X\cup\{e\})\ge v_i(X)$ for every $X\subseteq M\setminus\{e\}$.
Additionally, we define an item $e\in M$ as a \emph{chore} for agent $i\in N$ if $v_i(X\cup\{e\})\le v_i(X)$ for every $X\subseteq M\setminus\{e\}$, with at least one of these inequalities being strict.
An instance is said to be \emph{monotone} if every $e\in M$ is a good for any agent $i\in N$.
An instance is said to be \emph{doubly monotone} if every item $e\in M$ is either a good or a chore for any agent $i\in N$.

An \emph{allocation} is an ordered partition $\bA=(A_1,\dots,A_n)$ of $M$, i.e., $\bigcup_{i\in N}A_i=M$ and $A_i\cap A_j=\emptyset$ for any distinct $i,j\in N$.
In allocation $\bA$, each agent $i$ receives the items of bundle $A_i$.
A \emph{subsidy} is a non-negative real vector $\bp=(p_1,\dots,p_n)\in\mathbb{R}_+^N$, where $p_i$ is the amount of money given to agent $i\in N$.
In an allocation with a subsidy $(\bA,\bp)$, the utility of each agent $i$ is $v_i(A_i)+p_i$. 
An allocation with a subsidy $(\bA,\bp)$ is \emph{envy-free} if $v_i(A_i)+p_i\ge v_i(A_j)+p_j$ for any pair of agents $i,j\in N$.
Our goal is to find an envy-free allocation with a subsidy $(\bA,\bp)$ such that the total subsidy $\sum_{i\in N}p_i$ (or maximum subsidy $\max_{i\in N}p_i$) is minimized.

An allocation $\bA$ is called \emph{envy-free up to one item (EF1)} if, for all $i,j\in[n]$, it holds that $v_i(A_i\setminus X)\ge v_i(A_j\setminus X)$ for some $X\subseteq A_i\cup A_j$ with $|X|\le 1$.
Similarly, an allocation $\bA$ is called \emph{envy-free up to $k$ items (EF$k$)} if, for all $i,j\in[n]$, it holds that $v_i(A_i\setminus X)\ge v_i(A_j\setminus X)$ for some $X\subseteq A_i\cup A_j$ with $|X|\le k$.
It is known that an EF1 allocation always exists and can be found in polynomial time if the valuations are doubly monotone~\cite{LiptonMaMo04,bhaskar+2021} or $n=2$~\cite{BBB+2020}.

\subsection{Envy-free Structure of Subsidies}
An allocation $\bA$ is called \emph{envy-freeable} if there exists a subsidy vector $\bp$ such that $(\bA,\bp)$ is envy-free.
We call such a subsidy vector \emph{envy-eliminating}.
In this subsection, we describe the structure of envy-eliminating subsidies. 

We fix an allocation $\bA=(A_1,\dots, A_n)$.
Let $w\in\mathbb{R}^{[n]\times [n]}$ be a weight matrix such that $w_{i,j}=v_i(A_j)$ for each $i,j\in [n]$.
For a permutation $\sigma$ of $[n]$, let $\bA^\sigma=(A_{\sigma(1)},\dots,A_{\sigma(n)})$ and 
let $P^\sigma$ be the set of envy-eliminating subsidy vectors: 
\begin{align}
  P^\sigma=\left\{\bp\in\mathbb{R}_+^n\,\middle|\, 
  w_{i,\sigma(i)}+p_{\sigma(i)}\ge w_{i,j}+p_{j}~(\forall i,j\in[n])
  \right\}.
\end{align}
We call a permutation $\sigma$ \emph{maximum weight permutation} for $w$ if it maximizes $\sum_{i=1}^n w_{i,\sigma(i)}$ among permutations.
\citet{HalpernShah2019} proved that an allocation $\bA$ is envy-freeable if and only if $\sum_{i=1}^n v_i(A_i)\ge \sum_{i=1}^n v_i(A_{\sigma(i)})$ for any permutation $\sigma$ of $[n]$.
From this result, the polyhedron $P^\sigma$ is non-empty if $\sigma$ is a maximum weight permutation.
The polyhedron $P^\sigma$ contains a unique minimal element (i.e., a vector $\bp$ such that $\bp \leq \bp'$ for any $\bp' \in P^{\sigma}$) because it is lower bounded by non-negativity and $\bp,\bp'\in P^\sigma$ implies that $(\min\{p_i,p_i'\})_{i\in[n]}$ is also in $P^\sigma$.
The unique minimal vector is called the \emph{minimum subsidy vector} for $w$ with respect to $\sigma$.
Note that, for $\bp\in P^\sigma$, the subsidy $p_i$ is associated with the bundle $A_i$, not the agent $i$.

For the pair $(w,\sigma)$ of a weight $w$ and a permutation $\sigma$ of $[n]$, 
we define the \emph{envy graph} $G^{w,\sigma}=(V,E;\gamma)$ as a weighted directed complete graph in which each agent is a vertex (i.e., $V=[n]$), and each edge $(i,j)\in E~(=\{(i',j')\mid i',j'\in [n], i'\neq j'\})$ has a weight $\gamma_{i,j}=w_{i,\sigma(j)}-w_{i,\sigma(i)}$.
Note that $\gamma_{i,j}$ represents the envy from $i$ towards $j$ in allocation $\bA^\sigma$.
The minimum subsidy vector can be characterized by using this envy graph.
\begin{lemma}[{\citet[Theorem 2]{HalpernShah2019}}]\label{lem:longest-path}
For any maximum weight permutation $\sigma$, the minimum subsidy $p_{\sigma(i)}$ is the maximum length of any path in $G^{w,\sigma}$ starting at $i$.
\end{lemma}
It should be noted that the envy graph $G^{w,\sigma}$ does not contain any positive-weight directed cycle if $\sigma$ is a maximum weight permutation.
Hence, we only need to consider simple paths.
Although there may exist several maximum weight permutations, the following lemma states that the corresponding minimum subsidy vectors are identical.
\begin{lemma}\label{lem:bundle-subsidy}
Let $\sigma$ and $\sigma'$ be maximum weight permutations for $w$.
Also, let $\bp$ and $\bp'$ be the minimum subsidy vectors for $w$ with respect to $\sigma$ and $\sigma'$, respectively.
Then, $\bp=\bp'$.
\end{lemma}
\begin{proof}
It is sufficient to prove that $\bp\in P^{\sigma'}$ and $\bp'\in P^{\sigma}$.
In addition, by symmetry, it is sufficient to show only the former.

Define a vector $\bq\in\mathbb{R}^n$ as $q_i= w_{i,\sigma(i)}+p_{\sigma(i)}$ for each $i\in[n]$ 
and a weight $w'\in \mathbb{R}^{[n]\times[n]}$ as $w'_{i,j}=w_{i,j}+p_j-q_i$ for each $i,j\in[n]$.
By definition of $\bp$, we have 
$w'_{i,j}=(w_{i,j}+p_j)-(w_{i,\sigma(i)}+p_{\sigma(i)})\le 0~(\forall i,j\in[n])$ and
$w'_{i,\sigma(i)}=0~(\forall i\in[n])$.

For any permutation $\pi$ of $[n]$, the difference between the total weights of $w$ and $w'$ is
\begin{align}
\sum_{i\in [n]} w_{i,\pi(i)} - \sum_{i\in [n]} w'_{i,\pi(i)}
=\sum_{i\in[n]}p_i-\sum_{i\in[n]}q_i,
\end{align}
which is a constant independent of $\pi$.
Thus, 
$\sigma$ and $\sigma'$ are maximum weight permutations for $w'$, and hence the total weight of $\sigma'$ for $w'$ is 
$\sum_{i\in [n]}w'_{i,\sigma'(i)} = \sum_{i\in [n]}w'_{i,\sigma(i)} = 0$.
As $w'$ is nonpositive,
it holds that $w'_{i,\sigma'(i)}=0$ for every $i\in[n]$.
Thus, for any $i,j\in[n]$, we have
\begin{align}
  w_{i,\sigma'(i)}+p_{\sigma'(i)}
  &=w'_{i,\sigma'(i)}+q_i
  =q_i\\
  &\ge w'_{i,j}+q_i
  = w_{i,j}+p_j,
\end{align}
where the inequality holds by $w'_{i,j}\le 0$.
Hence, $\bp$ must be in $P^{\sigma'}$.
\end{proof}
From this lemma, the minimum subsidy vector is determined for an allocation without specifying a maximum weight permutation.

It should be noted that the subsidy vector $\bp$ and the vector $\bq\in\mathbb{R}^n$ defined in the proof of \Cref{lem:bundle-subsidy} can be viewed as dual variables of an assignment problem. 
Here, the primal of the assignment problem is
\begin{align}
\begin{array}{rll}
\max       & \sum_{i=1}^n\sum_{j=1}^nw_{i,j} x_{i,j} &\\[3pt]
\text{s.t.}& \sum_{j=1}^nx_{i,j}=1 &(\forall i\in [n]),\\[2pt]
           & \sum_{i=1}^nx_{i,j}\ge 1 &(\forall j\in [n]),\\[2pt]
           & x_{i,j}\ge 0          &(\forall i,j\in [n]),
\end{array}\label{eq:primalLP}
\end{align}
and the dual is
\begin{align}
\begin{array}{rll}
\min       & \sum_{i=1}^n q_i - \sum_{j=1}^n p_j &\\[3pt]
\text{s.t.}& w_{i,j}+p_j-q_i\le 0 &(\forall i,j\in [n]),\\
           &p_j\ge 0              &(\forall j\in[n]).
\end{array}\label{eq:primalLP}
\end{align}
By the complementary slackness, we have $x_{i,j}(q_i-w_{i,j}-p_j)=0$ for any $i,j\in[n]$ if $\bm{x}$ and $(\bp,\bq)$ are optimal solutions.

The maximum weight permutation can be computed in polynomial time via a maximum-weight bipartite perfect matching algorithm.
The minimum subsidy vector for $w$ can be computed in polynomial time via the Floyd--Warshall algorithm.

\section{Bounding Subsidy for EF1 Allocations}

In this section, we prove the following key lemma.
\begin{lemma}\label{lem:upper}
Let $w\in\mathbb{R}^{[n]\times[n]}$ be a weight matrix.
We denote the sequence of numbers in descending order of $(\max_{j\in[n]}(w_{i,j}-w_{i,i}))_{i\in[n]}$ as $(\beta_1,\dots,\beta_n)$, i.e., $\beta_1\ge \beta_2\ge\dots\ge \beta_n$.
Let $p^*$ be the minimum subsidy vector for $w$. 
Then, the $r$th largest value among $p^*_1,\dots,p^*_n$ is at most $\sum_{\ell=1}^{n-r}\beta_\ell$ for $r=1,2,\dots,n$.
\end{lemma}
We remark that the identical permutation $\id$ may not be a maximum weight permutation.
Thus, this lemma cannot be directly derived from \Cref{lem:longest-path} because $G^{w,\id}$ can contain a positive directed cycle that results in an infinitely long path.

For an EF1 allocation $\bA=(A_1,\dots,A_n)$, if we set the weight matrix as $w = (v_i(A_j))_{i,j\in [n]}$,
then the numbers $\beta_i~(i\in[n])$ in \Cref{lem:upper} are nonnegative and at most $1$ because $v_i(A_i)\ge v_{i}(A_j)-1$ for all $i,j\in[n]$.
Thus, the $r$th largest value in the minimum subsidy vector is at most $n-r$. 
Recall that an EF1 allocation $\bA$ can be found in polynomial time if the valuations are doubly monotone~\cite{bhaskar+2021} or $n=2$~\cite{BBB+2020}. Additionally, a maximum weight permutation $\sigma$ and the minimum subsidy vector can be computed in polynomial time.
Thus, from \Cref{lem:upper}, we can obtain the following theorem.
\begin{theorem}\label{thm:EF1-upper}
If the valuations are doubly monotone or $n=2$, there exists an envy-free allocation with a subsidy $(\bA,\bp)$ such that $\max_{i=1}^np_i\le n-1$ and $\sum_{i=1}^np_i\le \sum_{\ell=1}^{n}(n-\ell)=n(n-1)/2$.
Moreover, such an envy-free allocation with a subsidy can be computed in polynomial time.
\end{theorem}
\Cref{lem:upper} also implies that, even when the valuations are general and we only have an EF$k$ allocation, we can still derive an envy-free allocation with a subsidy of at most $k(n-1)$ per agent and $k\cdot n(n-1)/2$ in total.

We first observe that the bound of \Cref{thm:EF1-upper} cannot be improved even if the valuations are monotone and additive as long as an arbitrary EF1 allocation is considered.
\begin{example}\label{ex:n-1}
Consider an instance with $M=\{e_{i,j}\mid i\in[n],\,j\in[n+1]\}$.
The valuation of agent $i\in [n]$ for an item $e_{i',j}~(i'\in [n], j\in [n+1])$ is
\begin{align}
v_i(\{e_{i',j}\})=\begin{cases}
n/(n+1) & \text{if }i'=i,\\ 
1       & \text{if }i'=i-1,\\
0       & \text{otherwise}.
\end{cases}
\end{align}
Let $A_i=\{e_{i,1},\dots,e_{i,n+1}\}$ for each $i\in[n]$.
Then, $\bA=(A_1,\dots,A_n)$ is an envy-freeable EF1 allocation. 
Indeed, the weight matrix $w = (v_i(A_j))_{i,j\in [n]}$ is 
\begin{align}
\begin{pNiceMatrix}[first-row,first-col]
      &A_1    & A_2 & \dots & A_{n-1}& A_n \\
1     &n      &     &       &        &     \\
2     &n+1    & n   &       &        &     \\
\vdots&       & n+1 & \ddots&        &     \\
n-1   &       &     & \ddots& n      &     \\
n     &       &     &       & n+1    & n
\end{pNiceMatrix}.
\end{align}
It is not difficult to see that
the identity permutation $\id$ is a maximum weight one, and a path that visits $n, n-1, \dots, 2, 1$ in this order is the longest one in $G^{w,\id}$.
Thus, the minimum subsidy vector for $\bA$ is $\bp=(0,1,\dots,n-1)$. 
Hence, $\max_{i\in[n]}p_i=n-1$ and $\sum_{i\in[n]}p_i=n(n-1)/2$ hold.
\end{example}

In what follows, we prove \Cref{lem:upper}.
Let $\bA=(A_1,\dots,A_n)$ be an allocation and let $\bp^*$ be the minimum subsidy vector for the weight matrix $w = (v_i(A_j))_{i,j\in[n]}$.
Let $\sigma^*$ be a maximum weight permutation for $w$.
Note that the envy graph $G^{w,\sigma^*}$ may have an edge with a large weight.
Thus, the proof is not straightforward in such a case.

To prove the theorem, we use the following notations.
For each $i\in[n]$, we denote $q_i=\max_{j\in[n]}(v_i(A_j)+p^*_j)$, i.e., the maximum valuation for agent $i$ over bundles including subsidies, and $r_i=-(v_i(A_i)+p^*_i-q_i)$, i.e., the maximum envy for agent $i$.
Also, let us define
\begin{align}
  \hat{w}_{i,j}
  &=
  \begin{cases}
    v_i(A_i)+r_i&\text{if }i=j,\\
    v_i(A_j)&\text{if }i\ne j
  \end{cases}\label{eq:def-hatw}
\end{align}
for each $i,j\in[n]$.
Note that $r_i\ge 0$ and $\hat{w}_{i,j}\ge w_{i,j}$ for any $i,j\in[n]$.
We write $\hat{\bp}$ for the minimum subsidy vector for $\hat{w}$.
We demonstrate that $\hat{\bp}=\bp^*$ holds.

\begin{lemma}\label{lem:p=hatp}
The identity permutation is a maximum weight permutation for $\hat{w}$.
Moreover, the minimum subsidy vectors $\bp^*$ and $\hat{\bp}$ are the same.
\end{lemma}
\begin{proof}
Since $(A_{\sigma^*(1)},\dots,A_{\sigma^*(n)})$ with $\bp^*$ is envy-free by \Cref{lem:bundle-subsidy}, we have 
\begin{align}
w_{i,j}+p^*_{j}-q_i
&=v_i(A_j)+p^*_{j}-q_i\le 0
\label{eq:w-le}
\end{align}
for all $i,j\in[n]$, and
\begin{align}
w_{i,\sigma^*(i)}+p^*_{\sigma^*(i)}-q_i
=v_i(A_{\sigma^*(i)})+p^*_{\sigma^*(i)}-q_i=0 
\label{eq:w-eq}
\end{align}
for all $i\in[n]$.
By the definition of $\hat{w}$, we have 
\begin{align}
\begin{aligned}
\hat{w}_{i,j}+p^*_{j}-q_i=w_{i,j}+p^*_{j}&-q_i\le 0
\end{aligned}\label{eq:hatw-le}
\end{align}
for all $i,j\in[n]$ with $i\ne j$ and
\begin{align}
  &\hat{w}_{i,i}+p^*_{i}-q_i=w_{i,i}+p^*_{i}-q_i+r_i= 0, \label{eq:hatw-eq}
\end{align}
for all $i\in[n]$, 
where the inequality holds by \eqref{eq:w-le}.
In addition, for any $i\in[n]$ with $i\ne \sigma^*(i)$, we have
\begin{align}
  &\hat{w}_{i,\sigma^*(i)}+p^*_{\sigma^*(i)}-q_i=w_{i,\sigma^*(i)}+p^*_{\sigma^*(i)}-q_i=0 \label{eq:hatw-eq2}
\end{align}
by \eqref{eq:w-eq}.
Thus, for each $i\in[n]$, we obtain
\begin{align}
\hat{w}_{i,\sigma^*(i)}
=q_i-p^*_{\sigma^*(i)}
=w_{i,\sigma^*(i)}\label{eq:w-hatw}
\end{align}
since the first equality holds by \eqref{eq:hatw-eq} and \eqref{eq:hatw-eq2} and the second equality holds by \eqref{eq:w-eq}.

By \eqref{eq:hatw-le} and \eqref{eq:hatw-eq}, the total weight $\sum_{i=1}^n \hat{w}_{i,\sigma(i)}$ is at most $\sum_{i\in[n]}q_i-\sum_{j\in[n]}p_j^*$ for any permutation $\sigma$.
Thus, $\sigma^*$ and the identical permutation $\id$ are maximum weight permutations for $\hat{w}$ since the total weight of $\sigma^*$ and $\id$ for $\hat{w}$ are $\sum_{i\in[n]}q_i-\sum_{j\in[n]}p_i^*$ by \eqref{eq:hatw-eq} and \eqref{eq:hatw-eq2}.
Moreover, $\bp^*$ is an envy-eliminating subsidy for $\hat{w}$ (with respect to $\sigma^*$) since $\hat{w}_{i,\sigma^*(i)}+p^*_{\sigma^*(i)}=q_i\ge \hat{w}_{i,j}+p^*_j$ for any $i,j\in[n]$ by \eqref{eq:hatw-le}, \eqref{eq:hatw-eq}, and \eqref{eq:hatw-eq2}.

As $\bp^*$ is an envy-eliminating subsidy for $\hat{w}$, we have $\hat{\bp}\le \bp^*$.
To prove that $\hat{\bp}=\bp^*$, what is left is to show that $\hat{\bp}$ is an envy-eliminating subsidy vector for $w$.

Define $\hat{q}_i=\max_{j\in[n]}(\hat{w}_{i,j}+\hat{p}_j)$ for each $i\in[n]$.
Since $\hat{\bp}$ is the minimum subsidy vector for $\hat{w}$ with respect to $\sigma^*$ by \Cref{lem:bundle-subsidy}, we have
\begin{align}
  &\hat{w}_{i,\sigma^*(i)}+\hat{p}_{\sigma^*(i)}-\hat{q}_i=0 \quad (\forall i\in[n]).\label{eq:hatw'-eq2}
\end{align}
Thus, for each $i\in[n]$, we have
\begin{align}
    w_{i,\sigma^*(i)}+\hat{p}_{\sigma^*(i)}
    &=\hat{w}_{i,\sigma^*(i)}+\hat{p}_{\sigma^*(i)}
    =\hat{q}_i
    =\max_{j\in[n]}(\hat{w}_{i,j}+\hat{p}_j)
    \ge \max_{j\in[n]}(w_{i,j}+\hat{p}_j),
\end{align}
where the first equality holds by \eqref{eq:w-hatw}, 
the second equality holds by \eqref{eq:hatw'-eq2}, and 
the last inequality holds by $\hat{w}_{i,j}\ge w_{i,j}$ for any $i,j\in[n]$.
Therefore, $\hat{\bp}$ is an envy-eliminating subsidy vector for $w$, which completes the proof.
\end{proof}

By definition, the weight of each edge in $G^{\hat{w},\id}$ is at most that of the corresponding edge in $G^{w,\id}$ because 
\begin{align}
    \hat{w}_{i,j}-\hat{w}_{i,i}
    =w_{i,j}-(w_{i,i}+r_i)
    \le w_{i,j}-w_{i,i}
\end{align}
for any $i,j\in[n]$ with $i\ne j$.
By combining \Cref{lem:p=hatp} with \Cref{lem:longest-path}, we prove \Cref{lem:upper}.
\begin{proof}[Proof of \Cref{lem:upper}]
Recall that $\hat{\bp}=\bp^*$ and $\id$ is a maximum weight permutation
by \Cref{lem:p=hatp}.
For each $i\in[n]$, let $P_i\subseteq E$ be a longest path in $G^{\hat{w},\id}$ starting from $i$. 
By \Cref{lem:longest-path}, $\hat{p}_i~(=p^*_i)$ is the length of $P_i$ in $G^{\hat{w},\id}$.
Note that $P_i$ is a simple path.
As the weight of each edge in $G^{\hat{w},\id}$ is at most that of the corresponding edge in $G^{w,\id}$, we have 
\begin{align}
p^*_i
&=\sum_{(s,t)\in P_i}(\hat{w}_{s,t}-\hat{w}_{s,s})
\le\sum_{(s,t)\in P_i}(w_{s,t}-w_{s,s})
\le \sum_{\ell=1}^{|P_i|}\beta_\ell
\label{eq:p^* bound}
\end{align}
for each $i\in[n]$.

What is left is to provide upper bounds of the numbers of edges in the longest paths $P_1,\dots,P_n$.
Let $S = \bigcup_{i\in [n]}P_i$.
Without loss of generality, we may assume that, if two paths $P_i$ and $P_j$ share a common vertex, all of the edges that follow the vertex in these two paths are identical.
Then, $S$ is a directed forest and $|S|\le n-1$.
We relabel the vertices as $|P_1|\ge |P_2|\ge |P_n|$.
Then, $|P_1|\le n-1$ since $P_1\subseteq S$ and $|S|\le n-1$. 
As $P_1$ is longest among $P_1,\dots,P_n$, every path in $P_2,\dots,P_n$ does not use vertex $1$.
Indeed, if $P_i~(i>1)$ passes vertex $1$, then by the assumption, we have $|P_i|>|P_1|$.
Hence, $\bigcup_{i=2}^{n}P_i$ forms a directed forest that does not contain vertex $1$, and  $|P_2|\le n-2$.
By repeatedly applying the same argument, we have $|P_i|\le n-i$ for $i=1,2,\dots,n$.

Therefore, for $r=1,2,\dots,n$, the $r$th largest value in $p^*$ is at most $\sum_{\ell=1}^{|P_r|}\beta_{\ell}\le\sum_{\ell=1}^{n-r}\beta_\ell$ by \eqref{eq:p^* bound}.
\end{proof}

\section{Improved Bounds for Monotone Valuations}
In this section, we provide an improved upper bound of subsidy when the valuations are monotone.
As observed in \Cref{ex:n-1}, a maximum subsidy of $n-1$ is required to guarantee envy-freeness for an EF1 allocation.
We demonstrate that the upper bound can be improved by slightly modifying a given EF1 allocation.
Formally, we present the following theorem.
\begin{theorem}\label{thm:improved}
  If $n\ge 3$ and the valuations are monotone, there exists an envy-free allocation with a subsidy $(\bA,\bp)$ such that $\max_{i\in[n]}p_i\le n-1.5$ and $\sum_{i\in[n]}p_i\le (n^2-n-1)/2$.
  Moreover, such an envy envy-free allocation with a subsidy can be computed in polynomial time.
\end{theorem}

Note that if $n=2$, \Cref{thm:EF1-upper} implies that there exists an envy-free allocation with a subsidy where only one agent receives a subsidy of at most $1$. This bound cannot be improved even when there is one item with a value of $1$ for each agent.

In what follows, we assume that $n \geq 3$.
We describe that we can obtain in polynomial time an EF1 allocation $\bA$ satisfying $\sum_{i\in[n]}v_i(A_i)\ge \sum_{i\in[n]}v_i(A_{\sigma(i)})$ for any permutation $\sigma$ such that $\bA^{\sigma}$ is EF1.
We first compute an EF1 allocation $\bX$ in polynomial time using the envy-cycles algorithm~\cite{LiptonMaMo04}.
Next, we modify $\bX$ as follows.
Construct a bipartite graph $([n],[n]; E)$ where an edge $(i,j)\in E$ exists if and only if the EF1 criterion still holds for agent $i$ when we swap bundles of agents $i$ and $j$, i.e., $v_i(X_j)\ge \min_{Y\subseteq X_k:\,|Y|\le 1}v_i(X_k\setminus Y)$ for all $k\in[n]$. We assign the weight of an edge $(i,j)\in E$ as $v_i(X_j)$.
Then we find the maximum weight perfect matching on the bipartite graph.
By permutating bundles according to the matching, we can obtain the desired allocation $\bA$.

Define $w=(v_i(A_j))_{i,j\in [n]}$. 
Let $\bp^*$ be the minimum subsidy vector for $w$.
In addition, let $q_i=\max_j(v_i(A_j)+p^*_j)$ and $r_i=-(v_i(A_i)+p^*_i-q_i)$ for each $i\in[n]$.
Let $\hat{w}$ be the weights defined as \eqref{eq:def-hatw}.

A main task in the proof of \Cref{thm:improved} is to show that there exists an allocation with a subsidy $(\bA'',\bp'')$ such that $\max_{i\in[n]}p''_i \leq n-1.5$ by modifying $(\bA, \bp^*)$.
We assume that $\max_{i\in[n]}p_i^*> n-1.5$ since otherwise (i.e., $\max_{i\in[n]}p_i^*\le n-1.5$) we have $\sum_{i\in[n]}p_i^*\le\sum_{k=1}^{n-1}\min\{k,\,n-1.5\}=(n^2-n-1)/2$ by \Cref{lem:upper}.
Here, $\beta_i \leq 1~(\forall i\in [n])$ in the lemma since $\bA$ is EF1.
By \Cref{lem:p=hatp}, $\bp^*$ is the minimum subsidy vector for $\hat{w}$, and $\id$ is a maximum weight permutation for $\hat{w}$.
Since $r_i\ge 0~(\forall i\in[n])$, the weight of edge $(i,j)$ in $G^{\hat{w},\id}$ is 
\begin{align}
\hat{w}_{i,j}-\hat{w}_{i,i}
&=v_i(A_j)-(v_i(A_i)+r_i)
\le v_i(A_j)-v_i(A_i)
\le 1.\label{eq:edge_upper}
\end{align}
By \Cref{lem:longest-path}, the length of a longest path in $G^{\hat{w},\id}$ is $\max_{i\in[n]} p^*_i~(>n-1.5)$.
Since the weight of each edge is at most $1$ by \eqref{eq:edge_upper}, the longest path must contain all the vertices and $n-1$ positive weight edges.
Without loss of generality, we may assume that $(n,n-1,\dots,1)$ is the longest path (see \Cref{fig:envy-graph-path}).

Let $s_i=1+r_i+v_i(A_i)-v_i(A_{i-1})$ for $i=2,3,\dots,n$.
Note that $1-s_i$ is the positive weight of $(i,i-1)$, and
\begin{align}
0\le r_i\le s_i\leq 1.\label{eq:ri<si}
\end{align}
For each $i\in[n]$, the path $(i,i-1,\dots,1)$ must be a longest path starting at $i$ in $G^{\hat{w},\id}$. 
This is because, if there was a longer path starting at $i$, we could replace the subpath of the longest path $(n,n-1,\dots,1)$ starting from $i$ with this longer path, thereby creating a longer path, contradicting the assumption that $(n,n-1,\dots,1)$ is a longest path.
Thus, for each $i\in[n]$, it holds that $p_i^*=\sum_{j=2}^i (1-s_j)$.
In addition, since $(n,n-1,\dots,1)$ is a longest path, $\max_{i\in[n]}p_i^*$ is achieved by $i=n$.
Then since $\sum_{i=2}^n (1-s_i)=\max_{i\in[n]}p_i^*>n-1.5$, we have
\begin{align}
0\le \sum_{i=2}^nr_i\le \sum_{i=2}^n s_i<0.5. \label{eq:si_sum}
\end{align}

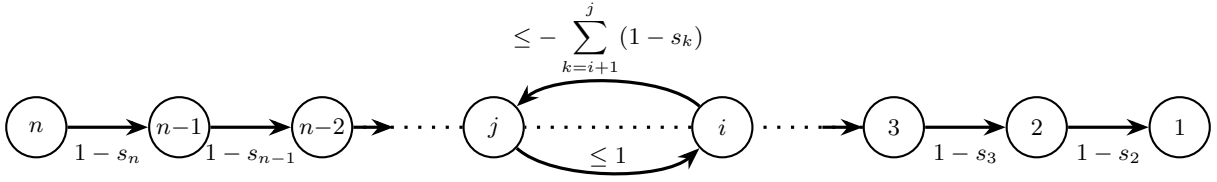
\begin{figure*}
\centering
\begin{tikzpicture}[xscale=1.9,font=\small]
\begin{scope}[every node/.style={circle,thick,draw,minimum size=8mm}]
    \node[label=center:$1$]     (A) at (8,0) {};
    \node[label=center:$2$]     (B) at (7,0) {};
    \node[label=center:$3$]     (C) at (6,0) {};
    \node[label=center:$i$]     (i) at (4.8,0) {};
    \node[label=center:$j$]     (j) at (3.2,0) {};
    \node[label=center:$n{-}2$] (D) at (2,0) {};
    \node[label=center:$n{-}1$] (E) at (1,0) {};
    \node[label=center:$n$]     (F) at (0,0) {};
\end{scope}
\draw[loosely dotted,very thick] (C) -- (i);
\draw[loosely dotted,very thick] (i) -- (j);
\draw[loosely dotted,very thick] (j) -- (D);
\begin{scope}[>={Stealth[black]},
              every node/.style={below=1mm},
              every edge/.style={draw=black,very thick}]
    \path [->] (B) edge node {$1-s_2$} (A);
    \path [->] (C) edge node {$1-s_3$} (B);
    \path [->] (5.5,0) edge node {} (C);
    \path [->] (D) edge node {} (2.5,0);
    \path [->] (E) edge node {$1-s_{n-1}$} (D);
    \path [->] (F) edge node {$1-s_n$} (E);
    \path [->] (i) edge[bend right=60] node[above,yshift=1pt,font=\small] {$\displaystyle \le -\sum_{k=i+1}^j(1-s_k)$} (j); 
    \path [->] (j) edge[bend right=60] node[above,yshift=1pt,font=\small] {$\le 1$} (i); 
\end{scope}
\end{tikzpicture}
\caption{The envy graph $G^{\hat{w},\id}$}
\label{fig:envy-graph-path}
\end{figure*}

Next, we observe that the weight of each edge, from a vertex with a lower index to a higher index, is small. This observation will be used to evaluate modified allocations.
\begin{lemma}\label{lem:what-upper}
For $i,j\in[n]$ with $i<j$, it holds that
\begin{align}
v_i(A_j)-v_i(A_i)-r_i
\le -\sum_{k=i+1}^{j}(1-s_k).
\end{align}
\end{lemma}
\begin{proof}
As the identity permutation is a maximum weight permutation for $\hat{w}$ by \Cref{lem:p=hatp}, envy graph $G^{\hat{w},\id}$ contains no positive-weight directed cycle. Hence, for $i,j\in[n]$ with $i<j$, we have
\begin{align}
(\hat{w}_{i,j}-\hat{w}_{i,i})
+\sum_{k=i+1}^{j}(\hat{w}_{k,k-1}-\hat{w}_{k,k})\le 0.
\end{align}
Since $\hat{w}_{k,k-1}-\hat{w}_{k,k} = 1-s_k$, we obtain
$v_i(A_j)-v_i(A_i)-r_i = \hat{w}_{i,j}-\hat{w}_{i,i} \le -\sum_{k=i+1}^{j}(1-s_k)$.
\end{proof}

For the case $i=1$, we can also obtain the following bound.
\begin{lemma}\label{lem:v1j-upper}
For every $j\in\{2,3,\dots,n\}$, we have
\begin{align}
v_1(A_j)\le \max\left\{v_1(A_1)-(1-s_2),\,\min_{e\in A_1}v_1(A_1\setminus\{e\})\right\}.
\end{align}
\end{lemma}
\begin{proof}
Let $j^*\in\argmax_{j\in\{2,3,\dots,n\}}v_1(A_j)$ and $\bA^{(j^*)}=(A_{j^*},A_1,A_2,\dots,A_{j^*-1},A_{j^*+1},\dots,A_n)$.
It suffices to prove that $v_1(A_{j^*})\le v_1(A_1)-(1-s_2)$ under the assumption that $v_1(A_{j^*})>\min_{e\in A_1}v_1(A_1\setminus\{e\})$.

In the allocation $\bA^{(j^*)}$, each agent $j \in \{2,3, \dots,n\}$ does not get worse than $\bA$ because $1-s_j \geq 0$, and thus the EF1 criterion is still satisfied for agent $j$.
By the choice of $j^*$, we have $v_1(A_{j^*})\ge v_1(A_j)$ for all $j\in\{2,3,\dots,n\}$
Hence, the allocation $\bA^{(j^*)}$ is EF1 if $v_1(A_{j^*})> \min_{e\in A_1}v_1(A_1\setminus\{e\})$.

Since $\bA^{(j^*)}$ is an EF1 allocation, we have $\sum_{j\in[n]}v_j(A_j)\ge \sum_{j\in[n]}v_j(A^{(j^*)}_j)$ by the definition of $\bA$.
This implies that 
\begin{align}
v_1(A_1)-v_1(A_{j^*})
&\ge \sum_{j=2}^{j^*}(v_j(A_{j-1})-v_j(A_{j}))
=\sum_{j=2}^{j^*}(1-s_j+r_j)
\ge 1-s_2,
\end{align}
by \eqref{eq:si_sum}.
Consequently, we obtain that 
$v_1(A_j)\le v_1(A_{j^*})\le v_1(A_1)-(1-s_2)$ for every $j\in\{2,3,\dots,n\}$.
\end{proof}

As $\bA$ is an EF1 allocation and $1-s_2 > 0$, we choose $e^*\in A_1$ such that $v_2(A_2)\ge v_2(A_1\setminus\{e^*\})$.
Define 
\begin{align}
\bA'=(A_1\setminus\{e^*\},A_2,A_3,\dots,A_{n-1},A_n\cup\{e^*\})
\end{align}
and let $w'$ be the weights such that $w'_{i,j}=v_i(A_j')~(i,j\in[n])$.
We demonstrate that the minimum subsidy vector $\bp'$ for $w'$ satisfies the conditions that $\max_{i\in[n]}p_i'\le n-1.5$.

Before we proceed to the proof, we observe the effect of this modification to the minimum subsidy vector for the instance in \Cref{ex:n-1}.
\begin{example}
Consider the instance observed in \Cref{ex:n-1}.
Then, $\bA$ in the example is an EF1 allocation such that $\sum_{i\in[n]}v_i(A_i)=n^2\ge\sum_{i\in[n]}v_i(A_{\sigma(i)})$ for any permutation $\sigma$.
Let $e^*=e_{1,1}$ and consider $\bA'=(A_1\setminus\{e^*\},A_2,A_3,\dots,A_{n-1},A_n\cup\{e^*\})$.
Then the valuations for the bundles are 
\begin{align}
(v_i(A'_j))_{i,j\in[n]}=
\begin{pNiceMatrix}[first-row,first-col]
      &A'_1      & A'_2& \dots & A'_{n-1}& A'_n \\
1     &\tfrac{n^2}{n+1}&     &       &        & \tfrac{n}{n+1}\\
2     &n         & n   &       &        &     \\
\vdots&          & n+1 & \ddots&        &     \\
n-1   &          &     & \ddots& n      &     \\
n     &          &     &       & n+1    & n
\end{pNiceMatrix}
\end{align}
and the minimum subsidy vector for this allocation is $(0,0,1,\dots,n-2)$.
\end{example}

Let us now proceed with the proof.
To provide upper bounds with \Cref{lem:upper}, we analyze the structure of $G^{w',\id}$.
\begin{lemma}\label{lem:w'-upper}
For each $i,j\in[n]$, the weight of edge $(i,j)$ in $G^{w',\id}$ is
\begin{align}
\MoveEqLeft
w'_{i,j}-w'_{i,i}
=v_i(A_j')-v_i(A_i')
\le\begin{cases}
\max\{0.5,\,v_1(A_n')-v_1(A_1')\} & \text{if }i=1,\\
0.5                               & \text{if }i=2,\\
1                                 & \text{if }i\geq 3.
\end{cases}
\end{align}
\end{lemma}
\begin{proof}
Let $i,j\in[n]$.
If $i=j$, the claim clearly holds as $w'_{i,j}-w'_{i,i}=0$.
Hence, we assume that $i\ne j$.

\medskip\noindent\textbf{Case 1. } $i=1$. 
If $2\le j<n$, 
we have
\begin{align}
w'_{i,j}-w'_{i,i}
&=v_1(A_j)-v_1(A_1')\\
&\le \max\left\{v_1(A_1)-(1-s_2),\ \min_{e\in A_1}v_1(A_1\setminus\{e\})\right\}-v_1(A_1\setminus\{e^*\})\\
&\le \max\{s_2,\,0\}=s_2< 0.5
\end{align}
by \Cref{lem:v1j-upper} and \eqref{eq:si_sum}.
If $j=n$, we have
\begin{align}
w'_{i,j}-w'_{i,i}=v_1(A_n')-v_1(A_1')
&\le \max\{0.5,\,v_1(A_n')-v_1(A_1')\}.
\end{align}

\medskip\noindent\textbf{Case 2. } $i=2$. 
If $j=1$, we have $w'_{i,j}-w'_{i,i}=v_2(A_1\setminus\{e^*\})-v_2(A_2)\le 0$ by the choice of $e^*$.
If $2<j<n$, we have
\begin{align}
w'_{i,j}-w'_{i,i}
&=v_2(A_j)-v_2(A_2)\\
&\le \hat{w}_{2,j}-\hat{w}_{2,2}+r_2
\le -\sum_{k=3}^j(1-s_k)+r_2\\
&\le -(1-s_3)+s_2
= -1+s_3+s_2\le 0
\end{align}
by \Cref{lem:what-upper}, \eqref{eq:ri<si} and \eqref{eq:si_sum}.
If $j=n$, we have
\begin{align}
w'_{i,j}-w'_{i,i}
&=v_2(A_n')-v_2(A_2)
\le v_2(A_n)-v_2(A_2)+1\\
&\le \hat{w}_{2,n}-\hat{w}_{2,2}+1+r_2
\le -\sum_{k=3}^n(1-s_k)+1+r_2
\le s_2+s_3 < 0.5.
\end{align}
by \Cref{lem:what-upper}, \eqref{eq:si_sum}, and $n\ge 3$.

\medskip\noindent\textbf{Case 3. } $i\ge 3$. 
If $j<n$, we have $w'_{i,j}-w'_{i,i}=v_i(A_j')-v_i(A_i')\le v_i(A_j)-v_i(A_i)\le 1$ by $A_i'\supseteq A_i$, $A_j'=A_j$, and~\eqref{eq:edge_upper}.
If $j=n$ (and hence $i\ne n$), we have
\begin{align}
w'_{i,j}-w'_{i,i}
&=v_i(A_n')-v_i(A_i)
\le v_i(A_n)-v_i(A_i)+1\\
&\le \hat{w}_{1,n}-\hat{w}_{1,1}+1+r_i
\le -\sum_{k=i+1}^n(1-s_k)+1+r_i\\
&\le -(1-s_n)+1+s_i
= s_i+s_n<0.5.
\end{align}
by \Cref{lem:what-upper} and \eqref{eq:si_sum}. 
\end{proof}

\Cref{lem:w'-upper} implies that each edge has a small weight except $(1,n)$.
In fact, there exist another allocation that induces small edge weights when the weight of $(1,n)$ is large.
Define 
$$\bA''=(A_n',A_1',A_2',\dots,A_{n-1}').$$
Let $w''$ be the weights such that $w''_{i,j}=v_i(A_j'')~(i,j\in[n])$.
Note that the minimum subsidy vector for $w''$ is $\bp''=(p'_n,p'_1,p'_2,\dots,p'_{n-1})$.
\begin{lemma}\label{lem:w''-upper}
For each $i,j\in[n]$, the weight of edge $(i,j)$ in $G^{w'',\id}$ is
\begin{align}
\MoveEqLeft[1] w''_{i,j}-w''_{i,i}
=v_i(A_j'')-v_i(A_i'')
\le\begin{cases}
0.5+\max\{0,\,v_1(A_1')-v_1(A_n')\} & \text{if }i=1,\\
s_2+s_3                             & \text{if }i= 2,\\
s_i                                 & \text{if }i\ge 3.
\end{cases}
\end{align}
\end{lemma}
\begin{proof}
Let $i,j\in[n]$. The proof is clear when $i=j$, and thus we assume that $i\neq j$.

\medskip\noindent\textbf{Case 1. } $i=1$. 
If $j=2$, we have 
\begin{align}
w''_{i,j}-w''_{i,i}=v_1(A_1')-v_1(A_n').
\end{align}
If $j>2$, we have
\begin{align}
w''_{i,j}-w''_{i,i}
&\le v_1(A_{j-1})-v_1(A_n')\\
&\le \max\left\{
v_1(A_1)-(1-s_2),\ 
\min_{e\in A_1}v_1(A_1\setminus\{e\})
\right\}-v_1(A_n')\\
&\le \max\{v_1(A_1')+s_2,\,v_1(A_1')\}-v_1(A_n')\\
&\le 0.5+ v_1(A_1')-v_1(A_n')
\end{align}
by \Cref{lem:v1j-upper} and \eqref{eq:si_sum}.

\medskip\noindent\textbf{Case 2. } $i=2$. 
In this case, since $r_2 \geq 0$, we have
\begin{align}
w''_{i,i}
=v_2(A_2'')
=v_2(A_1')
\ge v_2(A_1)-1
\geq v_2(A_2)-s_2.
\end{align}
If $j=1$, since $A''_1=A'_n=A_n\cup \{e^*\}$, we have
\begin{align}
w''_{i,j}-w''_{i,i}
&\le v_2(A_1'')-v_2(A_2)+s_2
= v_2(A_n')-v_2(A_2)+s_2\\
&\le 1+v_2(A_n)-v_2(A_2)+s_2\\
&\le -\sum_{k=3}^n(1-s_k)+1+s_2 \qquad (\because \text{\Cref{lem:what-upper}})\\
&\le -(1-s_3)+1+s_2=s_2+s_3.
\end{align}
If $j\ge 3$, we have 
\begin{align}
w''_{i,j}-w''_{i,i}
&\le v_2(A_{j-1}')-v_2(A_2)+s_2
= v_2(A_{j-1})-v_2(A_2)+s_2
\le -\sum_{k=3}^{j-1}(1-s_k)+s_2\le s_2.
\end{align}

\medskip\noindent\textbf{Case 3. } $i\ge 3$.
Recall that
$w''_{i,i}=v_i(A_i'')=v_i(A_{i-1})=v_i(A_i)+(1-s_i)$.

If $j=1$, we have
\begin{align}
w''_{i,j}-w''_{i,i}
&= v_i(A_1'')-v_i(A_i'')
= v_i(A_n\cup\{e^*\})-v_i(A_i)-(1-s_i)\\
&\le 1+v_i(A_n)-v_i(A_i)-(1-s_i)
= v_i(A_n)-v_i(A_i)+s_i\\
&\le -\sum_{k=i+1}^n (1-s_k)+r_i+s_i\\
&\le -1 + s_n +2s_i\le 0,
\end{align}
where the second inequality holds by \Cref{lem:what-upper} and the last inequality holds by \eqref{eq:si_sum}.
If $j\ge 2$, we have
\begin{align}
w''_{i,j}-w''_{i,i}
&= v_i(A_j'')-v_i(A_i'')
= v_i(A_{j-1}')-v_i(A_{i-1}')\\
&\le v_i(A_{j-1})-v_i(A_{i-1})
= v_i(A_{j-1})-v_i(A_i)-(1-s_i)\\
&\le 1-(1-s_i)=s_i,
\end{align}
where the last inequality holds by \eqref{eq:edge_upper}. 
\end{proof}

Now, we are ready to prove \Cref{thm:improved}.
\begin{proof}[Proof of \Cref{thm:improved}]
In what follows, suppose that $n\ge 3$.
If $\max_{i\in[n]}p_i^*\le n-1.5$, the total subsidy $\sum_{i\in[n]}p_i^*$ is at most $\sum_{k=1}^{n-1}\min\{k,\,n-1.5\}=(n^2-n-1)/2$ by \Cref{lem:upper}.
Therefore, we assume that $\max_{i\in[n]}p_i^*>n-1.5$.
We show that in this case, $\bp'$ and $\bp''$ defined before are desired ones.

For $\bp'$, by \Cref{lem:w'-upper,lem:upper}, we have
\begin{align}
\max_{i\in[n]}p_i'
&\le n-2+\max\{0.5,v_1(A_n')-v_1(A_1')\}\\
&\le (n-1.5)+\max\{0,v_1(A_n')-v_1(A_1')\}. \label{eq:max-p'}
\end{align}
and
\begin{align}
\sum_{i\in[n]}p_i'
&\le \sum_{i=1}^{n-2}(n-i)\cdot 1+1\cdot\max\{0.5,v_1(A_n')-v_1(A_1')\}\\
&\le \frac{n^2-n-1}{2}+\max\{0,v_1(A_n')-v_1(A_1')\}. \label{eq:sum-p'}
\end{align}
For $\bp''$, by \Cref{lem:w''-upper,lem:upper}, it holds that
\begin{align}
\max_{i\in[n]}p_i''
&\le s_3+\sum_{i=2}^n s_i+0.5+\max\{0,v_1(A_1')-v_1(A_n')\}\\
&\le 1.5+\max\{0,v_1(A_1')-v_1(A_n')\}\\
&\le (n-1.5)+\max\{0,v_1(A_1')-v_1(A_n')\}. \label{eq:max-p''}
\end{align}
and 
\begin{align}
\sum_{i\in[n]}p_i''
&\le \sum_{i=1}^{n-2}(n-i)\cdot \frac{1}{2}+\frac{1}{2}+\max\{0,v_1(A_n')-v_1(A_1')\}\\
&\le \frac{n^2-n-1}{2}+\max\{0,v_1(A_1')-v_1(A_n')\}. \label{eq:sum-p''}
\end{align}

We observe from these bounds that if $v_1(A_n')\le v_1(A_1')$, then $\bp'$ satisfies the requirements of this theorem, i.e., $\max_{i\in[n]}p_i'\le n-1.5$ by \eqref{eq:max-p'} and $\sum_{i\in[n]}p_i'\le (n^2-n-1)/2$ by \eqref{eq:sum-p'}; 
otherwise, $\bp''$ satisfies the requirements by \eqref{eq:max-p''} and \eqref{eq:sum-p''}.
However, recall that $\bA''$ and $\bp''$ are rearrangements of $\bA'$ and $\bp'$, respectively.
Thus, both $\bp'$ and $\bp''$ satisfy the requirements.
Hence, 
$(\hat{\bA}, \hat{\bp})=((A'_{\tau^*(1)},\dots,A'_{\tau^*(n)}),\,(p'_{\tau^*(1)},\dots,p'_{\tau^*(n)}))$ is a desired envy-free allocation with a subsidy, where $\tau^*$ is a maximum weight permutation for $w'$.
In addition, $(\hat{\bA}, \hat{\bp})$ can be computed in polynomial time via computing $\tau^*$.
\end{proof}

\section*{Acknowledgments}
This work was partially supported by the joint project of Kyoto University and Toyota Motor Corporation, titled ``Advanced Mathematical Science for Mobility Society'', 
JST PRESTO Grant Number
JPMJPR2122 
and 
JSPS KAKENHI Grant Numbers 
JP17K12646, 
JP19K22841, 
JP20H00609, 
JP20H05967, 
JP20K19739, 
JP20H00609, 
JP21H04979, 
JP21K17708, 
and JP21H03397. 

\bibliography{abb,fair}

\begin{thebibliography}{26}
\providecommand{\natexlab}[1]{#1}
\providecommand{\url}[1]{\texttt{#1}}
\expandafter\ifx\csname urlstyle\endcsname\relax
  \providecommand{\doi}[1]{doi: #1}\else
  \providecommand{\doi}{doi: \begingroup \urlstyle{rm}\Url}\fi

\bibitem[Alkan et~al.(1991)Alkan, Demange, and Gale]{Alkan1991}
A.~Alkan, G.~Demange, and D.~Gale.
\newblock Fair allocation of indivisible goods and criteria of justice.
\newblock \emph{Econometrica}, 59\penalty0 (4):\penalty0 1023--1039, 1991.

\bibitem[Aragones(1995)]{Aragones}
E.~Aragones.
\newblock A derivation of the money rawlsian solution.
\newblock \emph{Social Choice and Welfare}, 12:\penalty0 267--276, 1995.

\bibitem[Aziz(2021)]{Aziz2020}
H.~Aziz.
\newblock Achieving envy-freeness and equitability with monetary transfers.
\newblock In \emph{Proceedings of the 35th AAAI Conference on Artificial
  Intelligence (AAAI)}, pages 5102--5109, 2021.

\bibitem[Babaioff et~al.(2021)Babaioff, Ezra, and Feige]{babaioff2020fair}
M.~Babaioff, T.~Ezra, and U.~Feige.
\newblock Fair and truthful mechanisms for dichotomous valuations.
\newblock In \emph{Proceedings of the 35th AAAI Conference on Artificial
  Intelligence (AAAI)}, pages 5119--5126, 2021.

\bibitem[Barman et~al.(2022)Barman, Krishna, Narahari, and
  Sadhukhan]{BarmanKrishna2022IJCAI}
S.~Barman, A.~Krishna, Y.~Narahari, and S.~Sadhukhan.
\newblock Achieving envy-freeness with limited subsidies under dichotomous
  valuations.
\newblock In \emph{Proceedings of the 31st International Joint Conference on
  Artificial Intelligence (IJCAI)}, pages 60--66, 2022.

\bibitem[Barman et~al.(2023)Barman, Narayan, and Verma]{Barman2023}
S.~Barman, V.~V. Narayan, and P.~Verma.
\newblock Fair chore division under binary supermodular costs.
\newblock \emph{CoRR}, abs/2302.11530, 2023.
\newblock URL \url{https://arxiv.org/pdf/2302.11530.pdf}.

\bibitem[Benabbou et~al.(2021)Benabbou, Igarashi, Chakraborty, and
  Zick]{benabbou2020}
N.~Benabbou, A.~Igarashi, M.~Chakraborty, and Y.~Zick.
\newblock Finding fair and efficient allocations for matroid rank valuations.
\newblock \emph{ACM Transactions on Economics and Computation}, 9\penalty0
  (4):\penalty0 1--41, 2021.

\bibitem[Bhaskar et~al.(2021)Bhaskar, Sricharan, and Vaish]{bhaskar+2021}
U.~Bhaskar, A.~R. Sricharan, and R.~Vaish.
\newblock {On Approximate Envy-Freeness for Indivisible Chores and Mixed
  Resources}.
\newblock In \emph{Approximation, Randomization, and Combinatorial
  Optimization. Algorithms and Techniques (APPROX/RANDOM 2021)}, volume 207,
  pages 1:1--1:23, 2021.

\bibitem[Brustle et~al.(2020)Brustle, Dippel, Narayan, Suzuki, and
  Vetta]{Brustle2020}
J.~Brustle, J.~Dippel, V.~V. Narayan, M.~Suzuki, and A.~Vetta.
\newblock One dollar each eliminates envy.
\newblock In \emph{Proceedings of the 21st ACM Conference on Economics and
  Computation (EC)}, pages 23--39, 2020.

\bibitem[Budish(2011)]{Budi11a}
E.~Budish.
\newblock The combinatorial assignment problem: Approximate competitive
  equilibrium from equal incomes.
\newblock \emph{Journal of Political Economy}, 119\penalty0 (6):\penalty0
  1061--1103, 2011.

\bibitem[Bérczi et~al.(2020)Bérczi, Bérczi-Kovács, Boros, Gedefa, Kamiyama,
  Kavitha, Kobayashi, and Makino]{BBB+2020}
K.~Bérczi, E.~R. Bérczi-Kovács, E.~Boros, F.~T. Gedefa, N.~Kamiyama,
  T.~Kavitha, Y.~Kobayashi, and K.~Makino.
\newblock Envy-free relaxations for goods, chores, and mixed items, 2020.

\bibitem[Caragiannis and Ioannidis(2022)]{caragiannis2020computing}
I.~Caragiannis and S.~Ioannidis.
\newblock Computing envy-freeable allocations with limited subsidies.
\newblock In \emph{Proceedings of the 18th Conference on Web and Internet
  Economics (WINE)}, pages 522--539, 2022.

\bibitem[Foley(1967)]{Foley}
D.~K. Foley.
\newblock Resource allocation and the public sector.
\newblock \emph{Yale Economic Essays}, 7:\penalty0 45--98, 1967.

\bibitem[Goko et~al.(2022)Goko, Igarashi, Kawase, Makino, Sumita, Tamura,
  Yokoi, and Yokoo]{aamasversion}
H.~Goko, A.~Igarashi, Y.~Kawase, K.~Makino, H.~Sumita, A.~Tamura, Y.~Yokoi, and
  M.~Yokoo.
\newblock Fair and truthful mechanism with limited subsidy.
\newblock In \emph{Proceedings of the 21st International Conference on
  Autonomous Agents and Multi-Agent Systems (AAMAS)}, pages 534--542, 2022.

\bibitem[Halpern and Shah(2019)]{HalpernShah2019}
D.~Halpern and N.~Shah.
\newblock Fair division with subsidy.
\newblock In \emph{Proceedings of the 12th International Symposium on
  Algorithmic Game Theory (SAGT)}, pages 374--389, 2019.

\bibitem[Klijn(2000)]{Klijn}
F.~Klijn.
\newblock An algorithm for envy-free allocations in an economy with indivisible
  objects and money.
\newblock \emph{Social Choice and Welfare}, 17:\penalty0 201--215, 2000.

\bibitem[Lipton et~al.(2004)Lipton, Markakis, Mossel, and Saberi]{LiptonMaMo04}
R.~J. Lipton, E.~Markakis, E.~Mossel, and A.~Saberi.
\newblock On approximately fair allocations of indivisible goods.
\newblock In \emph{Proceedings of the 5th ACM Conference on Electronic Commerce
  (ACM-EC)}, pages 125--131, 2004.

\bibitem[Liu et~al.(2023)Liu, Lu, Suzuki, and Walsh]{liu2023mixed}
S.~Liu, X.~Lu, M.~Suzuki, and T.~Walsh.
\newblock Mixed fair division: A survey, 2023.

\bibitem[Maskin(1987)]{Maskin1987}
E.~S. Maskin.
\newblock \emph{On the fair allocation of indivisible goods}, pages 341--349.
\newblock Palgrave Macmillan UK, London, 1987.

\bibitem[Moulin(2004)]{moulin2004fair}
H.~Moulin.
\newblock \emph{Fair division and collective welfare}.
\newblock MIT press, 2004.

\bibitem[Narayan et~al.(2021)Narayan, Suzuki, and Vetta]{narayan2021birds}
V.~V. Narayan, M.~Suzuki, and A.~Vetta.
\newblock Two birds with one stone: Fairness and welfare via transfers.
\newblock In \emph{Proceedings of the 14th International Symposium on
  Algorithmic Game Theory (SAGT)}, pages 376--390, 2021.

\bibitem[Su(1999)]{SuRentalHarmony}
F.~E. Su.
\newblock Rental harmony: Sperner's lemma in fair division.
\newblock \emph{The American Mathematical Monthly}, 106\penalty0 (10):\penalty0
  930--942, 1999.

\bibitem[Sun and Yang(2003)]{SunYang2003}
N.~Sun and Z.~Yang.
\newblock A general strategy proof fair allocation mechanism.
\newblock \emph{Economics Letters}, 81\penalty0 (1):\penalty0 73--79, 2003.

\bibitem[Svensson(1983)]{Svensson1983}
L.-G. Svensson.
\newblock Large indivisibles: An analysis with respect to price equilibrium and
  fairness.
\newblock \emph{Econometrica}, 51\penalty0 (4):\penalty0 939--954, 1983.

\bibitem[Tadenuma and Thomson(1993)]{Tadenuma1993}
K.~Tadenuma and W.~Thomson.
\newblock The fair allocation of an indivisible good when monetary
  compensations are possible.
\newblock \emph{Mathematical Social Sciences}, 25\penalty0 (2):\penalty0 117 --
  132, 1993.

\bibitem[Varian(1974)]{Varian}
H.~R. Varian.
\newblock Equity, envy and efficiency.
\newblock \emph{Journal of Economic Theory}, 9:\penalty0 63 -- 91, 1974.

\end{thebibliography}

\end{document}